\documentclass[journal] {IEEEtran}

\makeatletter

\newcommand{\Rmnum}[1]{\expandafter\@slowromancap\romannumeral #1@}
\makeatother

\newtheorem{lemma}{Lemma}

\usepackage{booktabs}
\usepackage{diagbox}
\usepackage{multirow}
\usepackage{makecell}
\usepackage{longtable}
\usepackage{stfloats}
\usepackage{dsfont}
\usepackage{cite}
\usepackage{graphicx}
\usepackage{subfigure}
\usepackage[cmex10]{amsmath}
\usepackage{float}
\usepackage{array}
\usepackage{algorithm}
\usepackage{algpseudocode}
\usepackage{amsmath}
\usepackage{amsfonts}
\usepackage{amssymb}
\usepackage{ulem}
\usepackage{cancel}
\usepackage{amsmath,amssymb,amsfonts}
\usepackage{paralist,bbding,pifont}
\usepackage{subfigure}

\newtheorem{proof}{Proof}

\usepackage{xcolor}

\begin{document}
\title{Full-Duplex ISAC-Enabled D2D Underlaid Cellular Networks: Joint Transceiver Beamforming and Power Allocation}

\author{Tao Jiang,~Ming~Jin,~Qinghua~Guo, \IEEEmembership{Senior Member, IEEE}, Yinhong Liu, and  Yaming Li
\thanks{This work was partially supported by
Natural Science Foundation of China under Grant 61871246,
and Zhejiang Provincial Natural Science Funds under Grant LR21F010001.}
\thanks{T. Jiang, M. Jin, Y. Liu Y. Li are with the Faculty of Electrical Engineering and Computer Science, Ningbo University, Ningbo 315211, China (e-mail: 2011082045@nbu.edu.cn; jinming@nbu.edu.cn; 2211100052@nbu.edu.cn; 2211100133@nbu.edu.cn).}
\thanks{Q. Guo is with the School of Electrical, Computer and Telecommunications Engineering, University of Wollongong, Wollongong, NSW 2522, Australia (e-mail: qguo@uow.edu.au).}
\thanks{This work has been submitted to IEEE Transactions on Wireless Communications on 7 June.}}
\maketitle
\begin{abstract}
Integrating device-to-device (D2D) communication into
cellular networks can significantly reduce the transmission burden on base stations (BSs).
Besides, integrated sensing and communication (ISAC) is envisioned as a key feature in future wireless networks.
In this work, we consider a full-duplex ISAC-based D2D underlaid system, and propose a joint beamforming and power allocation scheme to improve the performance of the coexisting ISAC and D2D networks.
To enhance spectral efficiency, a sum rate maximization problem is formulated for the full-duplex ISAC-based D2D underlaid system, which is non-convex.
To solve the non-convex optimization problem, we propose a successive convex approximation (SCA)-based iterative algorithm and prove its convergence.
Numerical results are provided to validate the effectiveness of the proposed scheme with the iterative algorithm, demonstrating that the proposed scheme outperforms state-of-the-art ones in both communication and sensing performance.

 \end{abstract}
\section{Introduction}
\IEEEPARstart{W}ITH the explosive growth of the Internet of Things (IoT) devices, spectrum resource congestion is becoming increasingly prominent. Besides, many emerging applications in the sixth generation (6G), such as intelligent transportation and human-computer interaction, demand both high-quality communications and high-precision sensing simultaneously.
These practical needs motivate the design of integrated sensing and communication (ISAC) systems, which has attracted extensive attention from both industry and academia \cite{Liu2022}.
Moreover, advancements in technologies such as multiple-input multiple-output (MIMO) \cite{Fang2023},
signal processing technologies for ISAC \cite{Zhang2021} and hardware design
enable sensing and communication functionalities to be implemented using the same hardware and share same spectrum resource to achieve reliable communications and sensing, which could effectively improve spectral efficiency and energy efficiency.

In a MIMO ISAC system, beamforming plays a key role, where the spatial degrees of freedom (DoFs) provided by MIMO can be exploited to achieve substantial gains in both sensing and communications.
Many existing works focus on the design of transmit beamforming \cite{He2022,Wang2022,Li2023,Hua2023,Liu22022,LiuX2020,Ren2023},
with optimization problems centered around two key metrics for sensing:
1) Beampattern gains at the direction of targets; 2) Beampattern template matching.
Under the first metric, He {\it et al} in \cite{He2022} optimized the transmit beampattern by formulating an energy efficiency maximization problem while ensuring the beampattern gains of targets and the SINR requirements of communication users.
The transmit beampattern design was also considered in NOMA-ISAC systems \cite{Wang2022}, RIS-aided ISAC systems \cite{Li2023}, and secure ISAC systems \cite{Hua2023}.
In \cite{Wang2022}, a weighted objective function,  which involves the beampattern gain at a target and the SINR at communication users, is maximized.
In \cite{Li2023}, a reconfigurable intelligent surface (RIS) is employed to enhance the beampattern gain at the target direction by establishing an extra reflective link.
In a secure ISAC system in \cite{Hua2023}, the beampattern gain at the target (eavesdropper) is maximized while guaranteeing the secret rate and communication rate of legitimate users.
Under the second metric, in \cite{Liu22022}, the matching error between the designed and desired beampattern is minimized, incorporating the outage probability and power consumption constraints.
Based on \cite{Liu22022}, in \cite{LiuX2020} and \cite{Ren2023}, the authors considered a weighted minimization problem of matching error and cross-correlation among beampatterns of multiple targets.
However, these works focus on transmit beampatterns and ignore the reception of radar echo, which fails to capture the property of the receive signal at the BS.

In addition to transmit beamforming design, receive beamforming has been jointly used to filter the received signal and improve the signal to interference and noise ratio (SINR) for sensing at the base station (BS) \cite{Ashraf2023,LiuZ2023,He2023,Chu2023}.
The sensing SNR/SINR maximization problems were formulated in \cite{Ashraf2023},  corresponding to the perfect clutter removal and imperfect clutter removal scenarios, respectively.
In \cite{He2023}, the sensing SINR for target detection is used as a constraint in optimizing the power efficiency and spectral efficiency in a full-duplex ISAC system.
To suppress self-interference (SI) in a full-duplex ISAC system,
a joint ISAC transceiver beamforming design scheme was
proposed for maximizing a weighted objective function containing communication rate, transmit gain and receive gain \cite{LiuZ2023}.
Consistent with \cite{He2023}, the minimal sensing SINR was utilized as a constraint for minimizing the maximum eavesdrop rate of eavesdroppers in a secure DFRC system \cite{Chu2023}.

Recently, device-to-device (D2D) communication has attracted much attention.
It enables two nearby transceivers to establish a strong direct communication link without extra interventions of central agents, effectively reducing transmission delays, improving the spectral efficiency, and reducing power consumption simultaneously \cite{Lee2021}.
However, the interference among multiple D2D pairs are inevitable when they exploit the same spectral resource.
Power allocation plays a key role in managing interference and ensuring the quality of service (QoS) in D2D networks.
Weighted minimum mean squared error (WMMSE) is a classic iterative power allocation algorithm used for solving the weighted sum rate maximization problem in a centralized manner \cite{Christensen2008,Shi2011}.
However, WMMSE exhibits high complexity.
To tackle this issue, machine learning (ML) methods were exploited to reduce the computational complexity \cite{Chen2022,Gu2023}.
Combining the topology information, graph neural network (GNN) was adopted to implement power allocation among D2D pairs\cite{Chen2022}.
By modeling channel gains as the edge and node features, each D2D pair integrates neighbouring embeddings, and then update their own embedding in each GNN layer.
The power allocation result is obtained by classifying embedding, however it still operates in a centralized manner.
In \cite{Gu2023}, an improved GNN-based distributed power allocation scheme was proposed to reduce the CSI signalling overhead by broadcasting the pilot signal power and exploiting the temporal correlation of CSI.
However, the works mentioned above \cite{Christensen2008,Shi2011,Chen2022,Gu2023} only considered homogeneous D2D networks.

Actually, D2D pairs often communicate underlaying cellular networks,
which effectively alleviates the transmission burden of the central base station (BS) \cite{LiL2023,Vishnoi2023}.
Thus, power allocation is crucial for ensuring QoS for both cellular users (CUs) and D2D pairs.
Some studies, e.g., \cite{Xiao2020,Yang2016,Du2023}, consider a single-antenna BS and aim to maximize the energy efficiency \cite{Xiao2020} or spectral efficiency \cite{Yang2016}, and minimize the power consumption \cite{Du2023}, while ensuring the minimum communication SINR requirements.
Some works focus on multiple-antennas BS and leverage beamforming to enhance performance \cite{Cheng2022, Lin2016}.
In \cite{Cheng2022}, a robust beamforming scheme was proposed with given power for the D2D transmitters (D2D-TX).
In \cite{Lin2016}, the performance of transmit beamforming and power allocation were comprehensively analyzed in the coexistence of D2D  and downlink CU communications.
However, only a single D2D pair was considered.
Multiple CU downlink and D2D communications was discussed in \cite{WangW2022}, where an iterative algorithm is employed to solve the spectral efficiency maximization problem.
For the coexistence of multiple CU uplink and-D2D communications, a general communication and interference model was presented in \cite{Ozbek2020} to improve energy efficiency.
However, the works in \cite{Xiao2020,Yang2016,Du2023,Cheng2022,Lin2016,WangW2022,Ozbek2020}  consider a half-duplex BS.
In \cite{ChenC2020}, a full-duplex BS is adopted in a D2D underlaid cellular network.
However, it focuses on the issue of sub-channel allocation  and does not consider the issue of power allocation.

\subsection{Motivations and Contributions}\label{A}
ISAC and D2D underlaid cellular networks have generally progressed along separate, parallel paths despite extensive individual studies.
This work proposes integration of the full-duplex ISAC and the D2D underlaid cellular networks to improve the spectral efficiency,
and investigates the joint optimization of the the transceiver beamformings of the ISAC BS and the transmit power allocation of both the BS and D2D-TXs.
The main contributions of this paper are summarized as follows:
\begin{compactitem}
\item We propose a full-duplex ISAC and D2D underlaid cellular network, which involves a full-duplex BS, multiple D2D pairs and multiple CUs. To maintain the coexistence of the ISAC and D2D networks, we formulate a sum rate maximization problem while guaranteeing minimum sensing SINR for target detection, maximum transmit power  of D2D-TXs and the full-duplex BS.
\end{compactitem}
\begin{compactitem}
\item To deal with the non-convexity of the formulated problem, we employ the successive convex approximation (SCA) technology to covert the original problem into  a convex problem.
    The convergence and the computational complexity of the SCA-based algorithm are analyzed.
\end{compactitem}
\begin{compactitem}
\item Numerical results are provided to validate the effectiveness of the proposed scheme. Compared with the existing schemes, the propose scheme delivers the best performance in both communication and radar sensing.
\end{compactitem}

The remainder of the paper is organized as follows.
Section $\mathrm{\uppercase\expandafter{\romannumeral2}}$ presents the system models and provides the problem formulation under the criterion of the sum rate maximization.
In Section $\mathrm{\uppercase\expandafter{\romannumeral3}}$,
an SCA-based optimization algorithm is proposed to solve the formulated problem,
and the convergence and the computational complexity of the algorithm are analyzed.
Numerical simulations are provided in Section $\mathrm{\uppercase\expandafter{\romannumeral4}}$
and conclusions are drawn in Section $\mathrm{\uppercase\expandafter{\romannumeral5}}$.

  \textit{Notation:} Throughout the paper, boldface upper-case letters refer to matrices and boldface lower-case letters refer to vectors. $(\cdot)^H$, $(\cdot)^\mathrm{T}$, $(\cdot)^{-1}$, $\mathrm{Tr}(\cdot)$, $\mathbb{E}(\cdot)$ represent the conjugate transpose, transpose, inverse, trace and expectation operator, respectively. For a square matrix $\mathbf{D}$, $\mathbf{D}\succeq \mathbf{0}$ denotes that it is positive semidefinite, $\|\cdot\|$ denotes the Euclidean norm of a complex vector and $|\cdot|$ denotes the Euclidean norm of a complex number. $\mathcal{CN}(\mu,\sigma^2)$ denotes a complex Gaussian distribution with mean $\mu$ and variance $\sigma^2$.

\section{System Model and Problem Formulation}
\label{sec2}
\begin{figure}[!t]
\centerline{\includegraphics[width=2.8in]{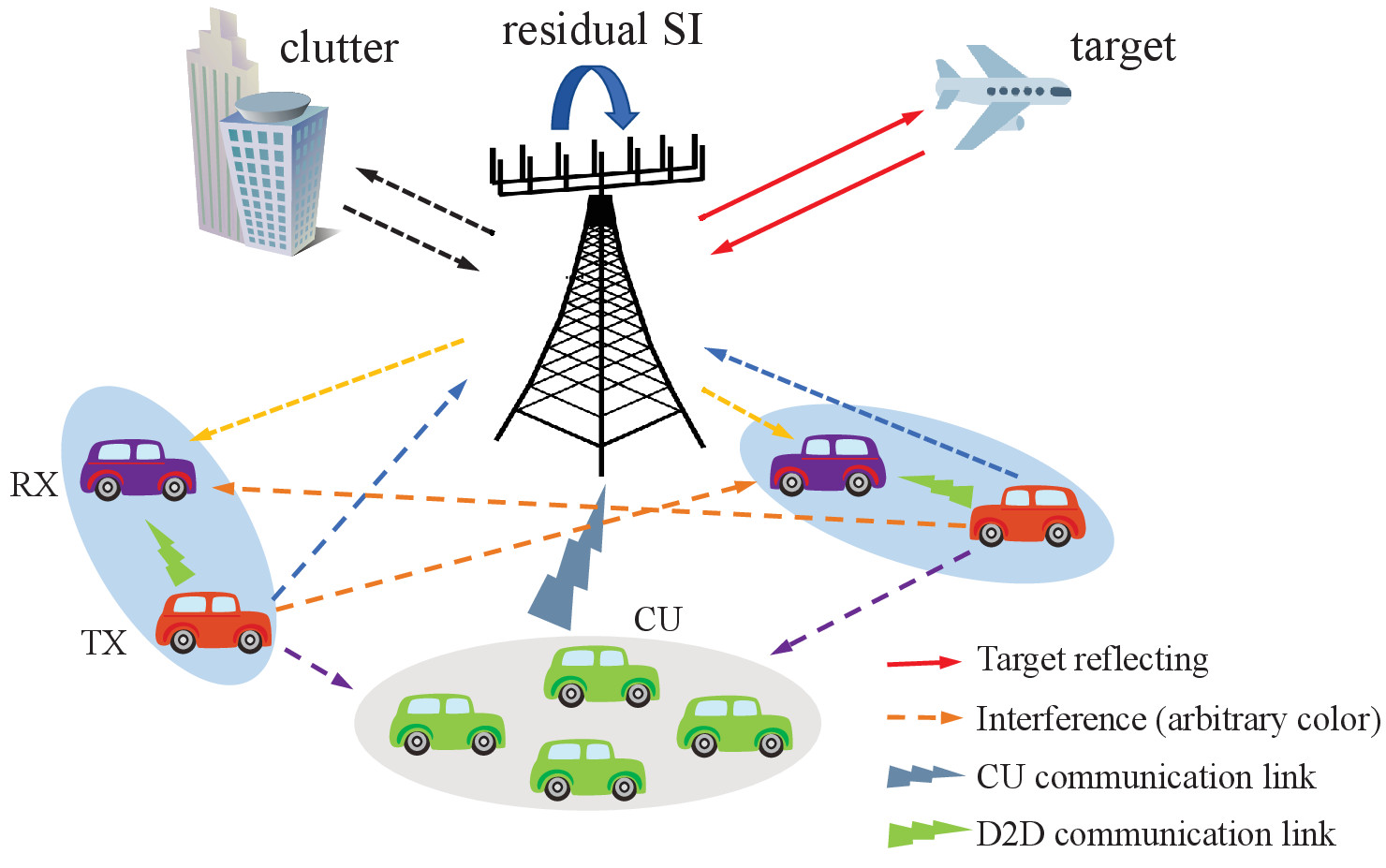}}
\caption{Illustration of an ISAC-empowered D2D underlaid cellular network.}
\label{model}
\end{figure}
We consider a joint system including full-duplex BS downlink communication, D2D communication, and radar sensing.
Specifically, as shown in Fig. \ref{model}, we assume an ISAC BS
equipped with $N_{t}$ transmit antennas and $N_{r}$ receive antennas.
Suppose that $K$ single-antenna CUs are served by the BS, denoted by  $\mathcal{K} = \{\mathrm{U}_{1},\dots,\mathrm{U}_{k}, \dots, \mathrm{U}_{K}\}$, and there are $M$ single-antenna D2D communication pairs
$\mathcal{M}=\left\{\left\{\text{D2D-TX}_{m},\text{D2D-RX}_{m}\right\}_{m=1}^{M}\right\}$, where D2D-TX$_{m}$ and D2D-RX$_{m}$ represent the transmitter and receiver of the $m$th D2D pair, respectively. Additionally, a target is at the angle of $\theta_{0}$, and $I$ potential signal-dependent clutters are at the angles of $\left\{\theta_{i}\right\}_{i=1}^{I},$ which reflect the signals from undesired directions and interfere with radar sensing.

\subsection{Communication Model}
\subsubsection{Downlink communication model}
The full-duplex BS transmits the ISAC signal $\mathbf{x}\in \mathbb{C}^{N_{t}\times 1}$ for downlink communication and radar sensing simultaneously, and it is given by
 \begin{equation}\label{eq01}
\mathbf{x} = \sum_{k=1}^{K} \mathbf{w}_{k}s_{k}^\text{CU}
\end{equation}
where $s_{k}^\text{CU}$ denotes the desired information symbol for the $k$th CU
and $\mathbf{w}_{k}\in \mathbb{C}^{N_{t}\times 1}$ denotes the corresponding beamforming vector.
Without loss of generality, it is assumed that information symbols are independent with each other and have unit power.

The received signal at the $k$th CU can be written as
\begin{equation}\label{eq02}
\begin{split}
 y_k&=\underbrace{\mathbf{h}_{\text{BS},  k}^H \mathbf{w}_k s_k^\text{CU}}_{\text{Desired~communication~signal}}\\
 &~~~~+\underbrace{\mathbf{h}_{\text{BS},k}^H \sum_{\substack{\forall k^{\prime} \neq k}} \mathbf{w}_{k^{\prime}} s_{k^{\prime}}^\text{CU}+\sum_{m=1}^M \sqrt{p_m} h_{m, k}^{\text{TX}} s_{m}^\text{D2D}+n_k}_{\text{Interference+Noise}}
 \end{split}
 \end{equation}
where $\mathbf{h}_{\text{BS}, k}\in \mathbb{C}^{N_{t}\times 1}$ and $h_{m, k}^{\text{TX}}$ represent the channel coefficients from the BS and D2D-TX$_{m}$ to $\mathrm{U}_{k}$, respectively,
$s_{m}^\text{D2D}$ and $p_{m}$ denote the information symbol and the corresponding transmit power of D2D-TX$_{m}$,
and $n_{k}$ denotes the additive circular symmetric complex white Gaussian noise with mean zero and variance $\sigma_{k}^{2}$, i.e., $n_{k}\sim\mathcal{CN}(0,\sigma_{k}^{2})$.
The corresponding SINR at the $k$th CU can be expressed as
\begin{equation}\label{eq03}
\begin{split}
 &\mathrm{SINR}_{k}^\text{CU}
=\frac{\left|\mathbf{h}_{\text{BS}, k}^H \mathbf{w}_k\right|^{2}}{\sum\limits_{\forall k^{\prime} \neq k}
\left|\mathbf{h}_{\text{BS}, k}^H \mathbf{w}_{k^{\prime}}\right|^{2}+\sum\limits_{m=1}^M p_{m} \left|h_{m, k}^{\mathrm{TX}} \right|^{2}+\sigma_{k}^{2}}.
 \end{split}
 \end{equation}
Thus, its achievable communication rate is given by
\begin{equation}\label{eq04}
\begin{split}
 C_{k}^\text{CU} = \mathrm{log}_{2}\left(1+\mathrm{SINR}_k^\text{CU}\right).
 \end{split}
 \end{equation}

\subsubsection{D2D communication model} The received signal at D2D-RX$_{m}$ contains the desired one from the corresponding transmitter D2D-TX$_{m}$, and the unexpected interferences from other D2D transmitters and the BS. Therefore, the overall received signal at D2D-RX$_{m}$ can be written as
\begin{equation}\label{eq05}
\begin{split}
 y_m&=\underbrace{\sqrt{p_{m}}h_{m, m}^{\text{TX}}s_{m}^\text{D2D}}_{\text{Desired~ communication~signal}}\\
 &~~+\underbrace{\sum_{\forall m^{\prime} \neq m}\sqrt{p_{m^{\prime}}}h_{m^{\prime}, m}^{\text{TX}}s^\text{D2D}_{m^{\prime}}+\mathbf{h}_{\text{BS}, m}^H \sum_{k=1}^K \mathbf{w}_{k} s_{k}^\text{CU}+n_{m}}_{\text{Interference+Noise}}
 \end{split}
 \end{equation}
where  $h_{m^{\prime}, m}^{\text{TX}}$ represents the channel coefficient from D2D-TX$_{m'}$ to D2D-RX$_{m}$,
$\mathbf{h}_{\text{BS}, m}\in \mathbb{C}^{N_{t}\times 1}$ denotes the channel coefficient between BS and D2D-RX$_{m}$, and $n_{m}\sim\mathcal{CN}(0,\sigma_{m}^{2})$ denotes the additive white Gaussian noise. Accordingly, the SINR at D2D-RX$_{m}$ is given by
\begin{equation}\label{eq06}
\begin{split}
 &\mathrm{SINR}_m^\text{D2D}=\frac{p_m\left|h_{m,m}^{\text{TX}}\right|^2}{\sum\limits_{\forall m^{\prime}\neq m}p_{m^{\prime}}\left|h_{m^{\prime}, m}^{\text{TX}}\right|^2+\sum\limits_{k=1}^K\left|\mathbf{h}_{\text{BS}, m}^H\mathbf{w}_k\right|^2+\sigma_{m}^2}.
 \end{split}
 \end{equation}
Accordingly, the achievable communication rate of the $m$th D2D pair is given by
\begin{equation}\label{eq07}
\begin{split}
 C_{m}^\text{D2D} = \log_{2}\left(1+\mathrm{SINR}_m^\text{D2D}\right).
 \end{split}
 \end{equation}
 \subsection{Radar Sensing Model}
It is assumed that the transmit and receive antennas at the BS are uniform linear arrays (ULAs) with half-wavelength inter-element spacing \cite{Wen2023}. The steering vector of transmit array $\mathbf{a}_t(\theta)$ and receive array $\mathbf{a}_r(\theta)$ are given by $\mathbf{a}_t(\theta)=\frac{1}{\sqrt{N_{t}}}\left[1,e^{j\pi\sin(\theta)},\ldots,e^{j\pi(N_{t}-1)\sin(\theta)}\right]^\mathrm{T}\in \mathbb{C}^{N_{t}\times 1}$ and  $\mathbf{a}_r(\theta)=\frac{1}{\sqrt{N_{r}}}\left[1,e^{j\pi\sin(\theta)},\ldots,e^{j\pi(N_{r}-1)\sin(\theta)}\right]^\mathrm{T}\in \mathbb{C}^{N_{r}\times 1}$, respectively.
In addition, we assume that a line-of-sight (LOS) channel exists between the BS and the target.
Thus, the reflected signal from the target at the angle of $\theta_{0}$ is given by $\alpha_{0}\mathbf{a}_{r}\left(\theta_{0}\right)\mathbf{a}_{t}^{H}\left(\theta_{0}\right)\mathbf{x}$,
where $\alpha_{0}$ denotes the complex amplitude of the target.
Similarly, the interference from clutters at angles $\{\theta_{i}\}_{i=1}^{I}$ is given by $\sum_{i=1}^{I}\alpha_{i}\mathbf{a}_{r}\left(\theta_{i}\right)\mathbf{a}_{t}^{H}\left(\theta_{i}\right)\mathbf{x}$,
where $\{\alpha_{i}\}_{i=1}^{I}$ denote the complex amplitudes of the clutters.
The residual SI is $\mathbf{H}_{\mathrm{SI}}\mathbf{x}$,
where $\mathbf{H}_{\mathrm{SI}}\in \mathbb{C}^{N_r\times N_t}$ denotes the residual SI channel between the transmit and receive antennas at the BS. Integrating the interference from  D2D-TXs, the received signal  at the BS is given by
 \begin{equation}\label{eq08}
\begin{split}
 \mathbf{y}_{r}&=\underbrace{\alpha_{0}\mathbf{a}_{r}\left(\theta_{0}\right)\mathbf{a}_{t}^{H}\left(\theta_{0}\right)
 \mathbf{x}}_{\text{Target~ reflecting}}+\underbrace{\sum_{i=1}^{I}\alpha_{i}\mathbf{a}_{r}\left(\theta_{i}\right)\mathbf{a}_{t}^{H}\left(\theta_{i}\right)\mathbf{x}}_{\text{Clutter~ reflecting}}\\
 &~~~~~~+\underbrace{\sum_{m=1}^{M}\sqrt{p_{m}}\mathbf{h}_{m, \text{BS}}^{\text{TX}}s_{m}^\text{D2D}}_{\text{Interference~from~D2D}\small{\text{-}}\mathrm{TXs}}+\underbrace{\mathbf{H}_{\mathrm{SI}}\mathbf{x}}_{\text{SI}}+\underbrace{\mathbf{n}_{r}}_{\text{Noise}}
 \end{split}
 \end{equation}
where $\mathbf{h}_{m, \text{BS}}^{\text{TX}}$ denotes the channel coefficient between D2D-TX$_{m}$ and BS,
$\mathbf{n}_{r}\sim \mathcal{CN}(0,\sigma_{r}^{2}\mathbf{I}_{N_{r}})$ denotes the additive white Gaussian noise,
and $\mathbf{I}_{N_{r}}\in \mathbb{R}^{N_{r}\times N_{r}}$ is the identity matrix.
With a receive beamformer $\mathbf{u}\in \mathbb{C}^{N_{r}\times 1}$  at the BS, we have
\begin{equation}\label{eq09}
\begin{split}
 y_{r}& =\mathbf{u}^{H}\mathbf{y}_{r}\\
 &=\alpha_{0}\mathbf{u}^{H}\mathbf{a}_{r}\left(\theta_{0}\right)\mathbf{a}_{t}^{H}\left(\theta_{0}\right)
 \mathbf{x}+\mathbf{u}^{H}\sum_{i=1}^{I}\alpha_{i}\mathbf{a}_{r}\left(\theta_{i}\right)\mathbf{a}_{t}^{H}\left(\theta_{i}\right)\mathbf{x}\\
 &~~~~~~+\mathbf{u}^{H}\sum_{m=1}^{M}\sqrt{p_{m}}\mathbf{h}_{m, \text{BS}}^{\text{TX}}s_{m}+\mathbf{u}^{H}\mathbf{H}_{\text{SI}}\mathbf{x}+\mathbf{u}^{H}\mathbf{n}_{r}.
\end{split}
\end{equation}
Denote $\alpha_{i}\mathbf{a}_{r}\left(\theta_{i}\right)\mathbf{a}_{t}^{H}\left(\theta_{i}\right)$
as $\mathbf{A}(\alpha_{i};\theta_{i})$, and let $\mathbf{B}=\sum_{i=1}^{I}\mathbf{A}(\alpha_{i};\theta_{i})+\mathbf{H}_\mathrm{SI}$. According to \cite{He2023}, the SINR of radar sensing is given by
 \begin{small}
 \begin{equation}\label{eq10}
 \begin{split}
&\mathrm{SINR}_{r}\\
&=\frac{\mathbb{E}\{|\mathbf{u}^{H}\mathbf{A}(\alpha_{0};\theta_{0})\mathbf{x}|^{2}\}}{\mathbb{E}\{|\mathbf{u}^{H}\mathbf{B}\mathbf{x}|^{2}\}+\sum\limits_{m=1}^{M}\mathbb{E}\{|\sqrt{p_{m}}\mathbf{u}^{H}\mathbf{h}_{m, \text{BS}}^{\text{TX}}s_{m}|^{2}\}+\mathbb{E}\{|\mathbf{u}^{H}\mathbf{n}_{r}|^{2}\}}. \\
\end{split}
 \end{equation}
 \end{small}By letting
\begin{equation}
  \mathbf{F}=\mathbb{E}\left\{\mathbf{xx}^{H}\right\}=\sum_{k=1}^{K}\mathbf{w}_{k}\mathbf{w}_{k}^{H},
\end{equation}
$\mathrm{SINR}_{r}$ can be rewritten as
  \begin{small}
 \begin{equation}\label{eq11}
 \begin{split}
 \mathrm{SINR}_{r}=\frac{\mathbf{u}^{H}\mathbf{A}(\alpha_{0};\theta_{0})\mathbf{F}\mathbf{A}^{H}(\alpha_{0};\theta_{0})\mathbf{u}}{\mathbf{u}^{H}\left(\mathbf{B}\mathbf{F}\mathbf{B}^{H}+\sum\limits_{m=1}^{M}p_{m}\mathbf{h}_{m, \text{BS}}^{\text{TX}}(\mathbf{h}_{m, \text{BS}}^{\text{TX}})^H+\sigma_{r}^{2}\mathbf{I}_{N_{r}}\right)\mathbf{u}}.
 \end{split}
 \end{equation}
 \end{small}

\subsection{Problem Formulation}
An optimization problem  is  formulated to maximize the sum rate of CU and D2D by jointly optimizing the transmit beamforming vectors $\{\mathbf{w}_{k}\}_{k=1}^{K}$ at the BS, the receive beamforming vector $\mathbf{u}$, and the power of D2D-TXs $\{p_{m}\}_{m=1}^{M}$,
considering the constrains on power consumption at the BS and D2D transmitters and the SINR requirement on target detection.

Use $P_{\mathrm{BS}}$ and $P_{m}$ to denote the maximum transmit power of the BS and D2D-TX$_{m}$, respectively.
Accordingly, we can formulate the sum rate maximization problem  as
\begin{subequations}\label{eq13}
\begin{align}
\underset{\mathbf{w}_{k},\mathbf{u},p_{m}}
{\text{max}}&\sum_{k=1}^{K}C_{k}^\text{CU}+\sum_{m=1}^{M}C_{m}^\text{D2D}\label{13a}\\
\text{s.t.}~~~~&~\mathrm{SINR}_{r}\geq \gamma_{r} \label{13b}\\
&\sum_{k=1}^{K}\|\mathbf{w}_{k}\|^{2}\leq P_{\mathrm{BS}}\label{13c}\\
&0\leq p_{m}\leq P_{m}, \forall m.\label{13d}
\end{align}
\end{subequations}
The global optimal solution to (\ref{eq13}) is intractable due to the following reasons: 1) The transmit beamformers $\{\mathbf{w}_{k}\}_{k=1}^{K}$, the transmit power of $\text{D2D-TX}_m$, i.e., $\{p_m\}_{m=1}^{M}$, and receive beamformer $\mathbf{u}$ are coupled in the constraint (\ref{13b}); 2) The objective function (\ref{13a}) is non-convex. Given this, solving (\ref{eq13}) is a challenging issue.

\section{Proposed Method}
\subsection{SCA-based Optimization}
\subsubsection{Optimization of $\mathbf{u}$}
Through examining the problem (\ref{eq13}), the receive beamforming vector $\mathbf{u}$ only exists in  radar SINR constraint (\ref{13b}). Therefore, it is reasonable to maximize the radar SINR concerning $\mathbf{u}$ to provide more freedom for optimizing other variables. The radar SINR maximization problem can be formulated as
\begin{equation}\label{eq14}
\begin{aligned}
\underset{\mathbf{u}}
{\max}~~\frac{\mathbf{u}^{H}\mathbf{A}(\alpha_{0};\theta_{0})\mathbf{F}\mathbf{A}^{H}(\alpha_{0};\theta_{0})\mathbf{u}}{\mathbf{u}^{H}\mathbf{G}\mathbf{u}}\\
\end{aligned}
\end{equation}
where
\begin{equation}\label{eq15}
  \mathbf{G} = \mathbf{B}\mathbf{F}\mathbf{B}^{H}+\sum_{m=1}^{M}p_{m}\mathbf{h}_{m, \text{BS}}^{\text{TX}}(\mathbf{h}_{m, \text{BS}}^{\text{TX}})^{H}+\sigma_{r}^{2}\mathbf{I}_{N_{r}}.
\end{equation}

\begin{lemma}
The close-form expression of the optimal $\mathbf{u}$ is given by
 \begin{equation}\label{eq16}
\begin{aligned}
\mathbf{u}^{*} = \frac{\mathbf{G}^{-1}\mathbf{a}_{r}\left(\theta_{0}\right)}{\mathbf{a}_{r}^{H}\left(\theta_{0}\right)\mathbf{G}^{-1}\mathbf{a}_{r}\left(\theta_{0}\right)}.
\end{aligned}
\end{equation}
\end{lemma}
\begin{proof}
See Appendix A.
\end{proof}
Substituting $\mathbf{u}^{*}$ into  (\ref{13b}) gives
 \begin{equation}\label{15}
\begin{aligned}
\mathrm{SINR}_{r}^{\ast} = |\alpha_{0}|^{2}\mathbf{a}_{t}^{H}\left(\theta_{0}\right)\mathbf{F}\mathbf{a}_{t}\left(\theta_{0}\right)\mathbf{a}_{r}^{H}\left(\theta_{0}\right)\mathbf{G}^{-1}\mathbf{a}_{r}\left(\theta_{0}\right).
\end{aligned}
\end{equation}

\subsubsection{Optimization of $\mathbf{w}_{k}$ and $p_m$}
Replacing (\ref{13b}) with $\mathrm{SINR}_{r}^{*}\geq \gamma_{r}$
and introducing an auxiliary variable
\begin{equation}\label{17-1}
  \mathbf{W}_{k}=\mathbf{w}_{k}\mathbf{w}_{k}^{H},
\end{equation}
Problem (\ref{eq13}) can be reformulated as

\begin{subequations}\label{eq19}
\begin{align}
&\underset{\mathbf{W}_{k},p_{m}}
{\text{min}}-\sum_{k=1}^{K}\mathrm{log}_{2}\left(1+\mathrm{SINR}_k^\text{CU}\right)-\sum_{m=1}^{M}\log_{2}\left(1+\mathrm{SINR}_m^\text{D2D}\right)\label{19a}\\
&~\text{s.t.}~~|\alpha_{0}|^{2}\mathbf{a}_{t}^{H}\left(\theta_{0}\right)\mathbf{F}\mathbf{a}_{t}\left(\theta_{0}\right)\mathbf{a}_{r}^{H}\left(\theta_{0}\right)\mathbf{G}^{-1}\mathbf{a}_{r}\left(\theta_{0}\right)
\geq \gamma_{r} \label{19b}\\
&~~~~~~~\sum_{k=1}^{K}\text{Tr}\left(\mathbf{W}_{k}\right)\leq P_{\mathrm{BS}}\label{19c}\\
&~~~~~~~0\leq p_{m}\leq P_{m}, \forall m\label{19d}\\
&~~~~~~~\mathbf{W}_{k}\succeq \mathbf{0}, \forall k\label{19f}\\
&~~~~~~~\mathrm{rank}(\mathbf{W}_{k})=1, \forall k. \label{19g}
\end{align}
\end{subequations}
where $\mathbf{F}$, $\mathrm{SINR}_{k}^\text{CU}$ and $\mathrm{SINR}_{m}^\text{D2D}$ can be rewritten as the functions of $\mathbf{W}_{k}$ rather than $\mathbf{w}_{k}$. In other words,
\begin{equation}\label{eq20}
\mathbf{F} = \sum_{k=1}^{K}\mathbf{W}_{k},
\end{equation}
\begin{equation}\label{eq21}
\begin{aligned}
&\mathrm{SINR}_{k}^\text{CU}=\frac{\mathbf{h}_{\text{BS}, k}^H\mathbf{W}_{k}\mathbf{h}_{\text{BS}, k}}{\sum\limits_{\forall k^{\prime}\neq k}\mathbf{h}_{\text{BS},k}^H\mathbf{W}_{k'}\mathbf{h}_{\text{BS}, k}+\sum\limits_{m=1}^Mp_m\left|h_{m, k}^{\text{TX}}\right|^2+\sigma_{k}^2},
\end{aligned}
\end{equation}
\begin{equation}\label{eq22}
\begin{aligned}
&\mathrm{SINR}_{m}^\text{D2D}=\frac{p_m\left|h_{m, m}^{\text{TX}}\right|^2}{\mathbf{h}_{\text{BS}, m}^H\mathbf{F}\mathbf{h}_{\text{BS}, m}+\sum\limits_{\forall m^{\prime}\neq m}p_{m'}\left|h_{m', m}^{\text{TX}}\right|^2+\sigma_{m}^2}.
\end{aligned}
\end{equation}
Problem (\ref{eq19}) is still a non-convex problem due to the object function, the constraint (\ref{19b}), and the rank one constraint (\ref{19g}).

We rewrite (\ref{19b}) as
 \begin{equation}\label{eq23}
\begin{aligned}
\frac{\gamma_{r}}{|\alpha_{0}|^{2}}\left(\mathbf{a}_{t}^{H}\left(\theta_{0}\right)\mathbf{F}\mathbf{a}_{t}\left(\theta_{0}\right)\right)^{-1}-\mathbf{a}_{r}^{H}\left(\theta_{0}\right)\mathbf{G}^{-1}\mathbf{a}_{r}\left(\theta_{0}\right)\leq 0.
\end{aligned}
\end{equation}
Let
\begin{equation}\label{eq24}
  \varphi=\mathbf{a}_{t}^{H}\left(\theta_{0}\right)\mathbf{F}\mathbf{a}_{t}\left(\theta_{0}\right).
\end{equation}
and we can obtain that $\varphi\geq 0$ since $\mathbf{F}\succeq \mathbf{0}$.
Moreover, it can also be obtained that  $\varphi$ is a linear function concerning $\mathbf{W}_{k}$.
Therefore, $\left(\mathbf{a}_{t}^{H}\left(\theta_{0}\right)\mathbf{F}\mathbf{a}_{t}\left(\theta_{0}\right)\right)^{-1}$  is convex with respect to  $\mathbf{W}_{k}$.
Similarly, we can also obtain that $\mathbf{a}_{r}^{H}\left(\theta_{0}\right)\mathbf{G}^{-1}\mathbf{a}_{r}\left(\theta_{0}\right)$ in (\ref{eq23}) is a convex function with respect to $\mathbf{W}_{k}$ and $p_{m}$ \cite{Boyd2004}.
Thus, (\ref{eq23}) is a difference-of-convex (DC) constraint.

We employ  SCA  to convert the DC constraint (\ref{eq23}) into a convex constraint.
Let
\begin{equation}\label{eq25}
  f(\mathbf{G}) = \mathbf{a}_{r}^{H}\left(\theta_{0}\right)\mathbf{G}^{-1}\mathbf{a}_{r}\left(\theta_{0}\right)
\end{equation}
and $\mathbf{G}^{(t-1)}$ be the solution at the $(t-1)$th iteration of the SCA.
The first-order approximation of Taylor expansion about $f(\mathbf{G})$ at around $\mathbf{G}^{(t-1)}$ is given by
\begin{equation}\label{eq26}
\begin{aligned}
&f(\mathbf{G}) \approx \widetilde{f}(\mathbf{G})=\mathbf{a}_{r}^{H}\left(\theta_{0}\right)\left(\mathbf{G}^{^{(t-1)}}\right)^{-1}\mathbf{a}_{r}\left(\theta_{0}\right)\\
&-\mathbf{a}_{r}^{H}\left(\theta_{0}\right)\left(\mathbf{G}^{^{(t-1)}}\right)^{-1}\left(\mathbf{G}-\mathbf{G}^{(t-1)}\right)\left(\mathbf{G}^{^{(t-1)}}\right)^{-1}\mathbf{a}_{r}\left(\theta_{0}\right)
\end{aligned}
\end{equation}
which is a linear function with respect to  $\mathbf{W}_{k}$ and $p_{m}$.
Hence,
\begin{equation}\label{eqn1}
  \frac{\gamma_{r}}{|\alpha_{0}|^{2}}\left(\mathbf{a}_{t}^{H}\left(\theta_{0}\right)\mathbf{F}\mathbf{a}_{t}\left(\theta_{0}\right)\right)^{-1}-
\widetilde{f}(\mathbf{G})\leq 0
\end{equation}
is a convex constraint. Moreover,  we can obtain that
\begin{equation}\label{eq27}
\begin{aligned}
&\widetilde{f}(\mathbf{G})\leq f(\mathbf{G}).
\end{aligned}
\end{equation}
Thus, the feasible domain of (\ref{eqn1}) is a subset of the feasible domain of (\ref{eq23}).

Next, we deal with the objective function (\ref{19a}).
The expression of the communication rate of the $k$th CU $C_{k}^\text{CU}$ can be rewritten as
\begin{equation}\label{eq28}
\begin{aligned}
-C_{k}^\text{CU} &= -\mathrm{log}_{2}\left(1+\mathrm{SINR}_{k}^\text{CU}\right)\\
&=-\mathrm{log}_{2}\left(\mathbf{h}_{\text{BS},k}^H\mathbf{F}\mathbf{h}_{\text{BS}, k}+\sum_{m=1}^Mp_m\left|h_{m, k}^{\text{TX}}\right|^2+\sigma^2_{k}\right)\\
&+\underbrace{\mathrm{log}_{2}\left(\sum_{\forall k^{\prime}\neq k}\mathbf{h}_{\text{BS}, k}^H\mathbf{W}_{k'}\mathbf{h}_{\text{BS}, k}+\sum_{m=1}^Mp_m\left|h_{m, k}^{\text{TX}}\right|^2+\sigma^2_{k}\right)}_{E_{k}^\text{CU}}.
\end{aligned}
\end{equation}
The non-convexity of (\ref{eq28}) lies in the second concave term $E_{k}^\text{CU}$.
According to SCA, the first-order approximation of Taylor expansion about $E_{k}^\text{CU}$, i.e.,  $\widetilde{E}_{k}^\text{CU}(t)$ at the $t$th iteration is given in (\ref{eq29}) at the bottom of the next page.
It can be obtained that $\widetilde{E}_{k}^\text{CU}(t) \geq E_{k}^\text{CU}$.

Similarly, we can rewrite the expression of communication rate of the $m$th D2D pair  as
\setcounter{equation}{30}
 \begin{equation}\label{eq30}
\begin{aligned}
-C_{m}^\text{D2D}& = -\mathrm{log}_{2}(1+\mathrm{SINR}_{m}^\text{D2D})\\
&=-\log_{2}\left(\mathbf{h}_{\text{BS},m}^H\mathbf{F}\mathbf{h}_{\text{BS},m}+\sum_{m'=1}^Mp_{m'}\left|h_{m', m}^{\text{TX}}\right|^2+\sigma^2_m\right)\\
&+\underbrace{\log_{2}\left(\mathbf{h}_{\text{BS}, m}^H\mathbf{F}\mathbf{h}_{\text{BS}, m}+\sum_{m^{\prime}\neq m}p_{m'}\left|h_{m', m}^{\text{TX}}\right|^2+\sigma^2_m\right)}_{E_{m}^\text{D2D}}
\end{aligned}
\end{equation}
and the non-convex term $E_{m}^\text{D2D}$ at the $t$th iteration is approximated with  the first-order Taylor expansion as in (\ref{eq31}), given at the bottom of the next page. It can also be obtained that $\widetilde{E}_{m}^\text{D2D}(t) \geq E_{m}^\text{D2D}$.

\newcounter{MYtempeqncnt2}
\setcounter{MYtempeqncnt2}{\value{equation}}
\setcounter{equation}{29}
\begin{figure*}[!hb]
\hrulefill
\normalsize
\begin{equation}\label{eq29}
 \begin{aligned}
\widetilde{E}_{k}^\text{CU}(t)&=\log_{2}\left(\sum_{k^{\prime}\neq k}\mathbf{h}_{\text{BS}, k}^H\mathbf{W}_{k'}^{(t-1)}\mathbf{h}_{\text{BS}, k}+\sum_{m=1}^Mp_m^{(t-1)}\left|h_{m, k}^{\text{TX}}\right|^2+\sigma^2\right)\\
&~~~~~~+\frac{\sum\limits_{k'\neq k}\text{Tr}\Big(\mathbf{h}_{\text{BS}, k}\mathbf{h}_{\text{BS}, k}^{H}\Big(\mathbf{W}_{k'}-\mathbf{W}_{k'}^{(t-1)}\Big)\Big)+\sum\limits_{m=1}^{M}\left|{h}_{m, k}^{\text{TX}}\right|^{2}\left({p}_{m}-{p}_{m}^{(t-1)}\right)}{\ln2\left(\sum\limits_{k^{\prime}\neq k}\mathbf{h}_{\text{BS}, k}^{H}\mathbf{W}_{k'}^{(t-1)}\mathbf{h}_{\text{BS}, k}+\sum\limits_{m=1}^{M}p_{m}^{(t-1)}\left|h_{m, k}^{\text{TX}}\right|^{2}+\sigma^{2}\right)}
\end{aligned}
\end{equation}
\setcounter{equation}{\value{MYtempeqncnt2}}
\end{figure*}
\newcounter{MYtempeqncnt3}
\setcounter{MYtempeqncnt3}{\value{equation}}
\setcounter{equation}{31}
\begin{figure*}[!hb]
\normalsize
\begin{equation}\label{eq31}
 \begin{aligned}
\widetilde{E}_{m}^\text{D2D}(t)&=\log_{2}\left(\sum_{k=1}^K\mathbf{h}_{\text{BS}, m}^H\mathbf{W}_{k}^{(t-1)}\mathbf{h}_{\text{BS}, m}+\sum_{m^{\prime}\neq m}p_{m'}^{(t-1)}\left|h_{m', m}^{\text{TX}}\right|^2+\sigma^2\right)\\
&~~~~~~+\frac{\sum\limits_{k=1}^{K}\text{Tr}\Big(\mathbf{h}_{\text{BS}, m}\mathbf{h}_{\text{BS}, m}^{H}\Big(\mathbf{W}_{k}-\mathbf{W}_{k}^{(t-1)}\Big)\Big)+\sum\limits_{m^{\prime}\neq m}\left|{h}_{m', m}^{\text{TX}}\right|^{2}\left({p}_{m'}-{p}_{m'}^{(t-1)}\right)}{\ln2\left(\sum\limits_{k=1}^{K}\mathbf{h}_{\text{BS}, m}^{H}\mathbf{W}_{k}^{(t-1)}\mathbf{h}_{\text{BS}, m}+\sum\limits_{m^{\prime}\neq m}p_{m'}^{(t-1)}\left|h_{m', m}^{\text{TX}}\right|^{2}+\sigma^{2}\right)}
\end{aligned}
\end{equation}
\setcounter{equation}{\value{MYtempeqncnt3}}
\end{figure*}

Therefore, we have the upper bounds $-\widetilde{C}_{k}^\text{CU}(t)$ and $-\widetilde{C}_{m}^\text{D2D}(t)$ for $-C_{k}^\text{CU}$ and $-C_{m}^\text{D2D}$, respectively,
where
\setcounter{equation}{32}
\begin{equation}\label{eq32}
\begin{aligned}
-\widetilde{C}_{k}^\text{CU}(t)=& -\mathrm{log}_{2}\left(\mathbf{h}_{\text{BS},k}^H\mathbf{F}\mathbf{h}_{\text{BS},k}+\sum_{m=1}^Mp_m\left|h_{m, k}^{\text{TX}}\right|^2+\sigma^2\right)\\
&~~~~+\widetilde{E}_{k}^\text{CU}(t).
\end{aligned}
\end{equation}
and
\begin{equation}\label{eq33}
\begin{aligned}
-\widetilde{C}_{m}^\text{D2D}(t)= &-\log_{2}\left(\mathbf{h}_{\text{BS}, m}^H\mathbf{F}\mathbf{h}_{\text{BS},m}+\sum_{m'=1}^Mp_{m'}\left|h_{m', m}^{\text{TX}}\right|^2+\sigma^2\right)\\
&~~~~+\widetilde{E}_{m}^\text{D2D}(t).
\end{aligned}
\end{equation}
After omitting the rank-one constraint (\ref{19g}), we have the SCA-based optimization problem at the $t$th iteration as
\begin{subequations}\label{eq34}
\begin{align}
\underset{\mathbf{W}_{k},p_{m}}
{\text{min}}&-\sum_{k=1}^{K} \widetilde{C}_{k}^\text{CU}(t) - \sum_{m=1}^{M} \widetilde{C}_{m}^\text{D2D}(t)\label{31a}\\
~\text{s.t.}~~~~&\frac{\gamma_{r}}{|\alpha_{0}|^{2}}\left(\mathbf{a}_{t}^{H}\left(\theta_{0}\right)\mathbf{F}\mathbf{a}_{t}\left(\theta_{0}\right)\right)^{-1}-
\widetilde{f}(\mathbf{G})\leq 0\label{31b}\\
&\sum_{k=1}^{K}\text{Tr}\left(\mathbf{W}_{k}\right)\leq P_{\mathrm{BS}}\label{31c}\\
&0\leq p_{m}\leq P_{m}, \forall m.\label{31d}\\
&\mathbf{W}_{k}\succeq \mathbf{0}, \forall k
\end{align}
\end{subequations}
which is a convex problem, and can be solved with convex optimization tools such as CVX.
The beamforming vector $\mathbf{w}_{k}$ can be obtained through
$\mathbf{W}_{k}$ with eigenvalue decomposition (EVD) or the Gaussian randomization method \cite{Luo2010}.

\subsection{Analysis of Convergence and Complexity}
Firstly, we analyze the convergence of the SCA-based optimization of (\ref{eq34}).
Denote the objective function of (\ref{eq34}) at the $t$th iteration by $opt_{1}\left(\mathbf{W}_{k}^{(t)},p_{m}^{(t)}\right)$,
where $\left\{\mathbf{W}_{k}^{(t)},p_{m}^{(t)}\right\}$ denotes the optimal solution at the $t$th iteration.
In SCA, the optimal solution $\left\{\mathbf{W}_{k}^{(t-1)},p_{m}^{(t-1)}\right\}$ at the $(t-1)$th iteration is also a feasible solution at the $t$th iteration \cite{Beck2010}.
Therefore, it can be obtained that $opt_{1}\left(\mathbf{W}_{k}^{(t)},p_{m}^{(t)}\right)\leq opt_{1}\left(\mathbf{W}_{k}^{(t-1)},p_{m}^{(t-1)}\right)$.
Moreover, the objective function of (\ref{eq34}) has a lower bound, leading to the convergence of the SCA-based optimization.

Now we analyse the complexity. According to \cite{Qi2023}, the worst-case complexity for solving a mix semidefinite and second-order cone programming problem is

\begin{equation}\label{eq35}
\begin{split}
\sqrt{\mu}\left(l\sum_{i=1}^{N_{\mathrm{soc}}}
\left(n_{i}^{\mathrm{SOC}}\right)^2+l^2\sum_{i=1}^{N_{\mathrm{sd}}}\left(n_{i}^{\mathrm{sd}}\right)^2
+l\sum_{i=1}^{N_{\mathrm{sd}}}\left(n_{i}^{\mathrm{sd}}\right)^3+l^3\right)\\
\times \ln{1/\epsilon}
\end{split}
\end{equation}
where $l$ is the number of variables, $N_{\mathrm{soc}}$ and $N_{\mathrm{sd}}$ are the number of second-order cone and semidefinite constraints, and $n_{i}^{\mathrm{SOC}}$ and $n_{i}^{\mathrm{sd}}$ are the dimensions of the $i$th second-order and semidefinite cones, $\mu=\sum_{i=1}^{N_{\mathrm{sd}}}n_{i}^{\mathrm{sd}}+2N_{\mathrm{soc}}$ is the barrier parameter, and $\epsilon>0$ is the solution precision.

The main complexity of each iteration lies in solving a semidefinite programming problem (\ref{eq34}). For (\ref{eq34}),
there are total $l=KN_t^2+M$ variables, $2M+1$ semidefinite constraints of size 1,
$K$ semidefinite constraints of size $N_t$, and $1$ semidefinite constraint of size 2.
Considering $I_{0}$ iterations, the complexity of the SCA-based optimization for (\ref{eq34}) is
$\mathcal{O}\bigg(I_{0}\sqrt{(2M+3+KN_{t})}\big[(KN_{t}^2+M)^2(5+2M+KN_{t}^2)+(KN_{t}^2+M)(9+2M+KN_{t}^3)+(KN_{t}^2+M)^3\big]\bigg)$ with $I_{0}$ iterations.

\section{Numerical Results}
\label{sec6}

In this section, we provide extensive numerical results to validate the effectiveness of the proposed ISAC design in a D2D underlaid cellular network.
In the simulations, the numbers of transmit antennas and receive antennas at the SB are set to $N_t=N_r =32$.
A target of interest is located at $\theta_{0} = 0^\circ$, and two clutters are located at $\theta_{1} = -50^\circ$ and $\theta_{2} = 40^\circ$.
The noise power at CUs, D2D receivers, and the BS are set to $\sigma_k^2 = \sigma_m^2 = \sigma_r^2 = -90\mathrm{dBm}$ \cite{Hua2023},
and the channel power gains of target and clutters are set to $|\alpha_0|^2=10\sigma_r^2$ and $|\alpha_1|^2=|\alpha_2|^2=10^3\sigma_r^2$ \cite{Tsinos2021}, respectively.
The $(i,j)$th element in residual SI channel coefficient matrix $(\mathbf{H}_\mathrm{SI})_{i,j}=\sqrt{\beta_{i,j}}\exp^{-j2\pi d_{i,j}/\lambda}$,
where $\beta_{i,j}$ denotes the residual SI channel power, $d_{i,j}$ denotes the  distance between the $i$th transmit antenna and the $j$th receive antenna. We let $\beta_{i,j}=-130\mathrm{dB}$, and model the $\exp^{-j2\pi d_{i,j}/\lambda}$ as a random variable with unit power and random phase \cite{Temiz2022}.
Additionally, we establish a two-dimensional Cartesian coordinates with the BS at the origin point,
and we use a tuple $(\theta, d)$ to represent the direction and distance information of CUs, D2D-TXs, and D2D-RXs, where $\theta$ denotes the angle  and $d$ denotes the distance between the users and the BS.
Assume that the BS serves $K=2$ CUs located at $(-75^\circ, 50\text{m})$ and $(20^\circ, 60\text{m})$.
Let $M=2$ pairs $\big\{\{$D2D-TX$_{m}$,D2D-RX$_{m}\}\big\}_{m=1,2}$ where D2D-RX$_1$ is randomly deployed within the range of  $\theta \in \left[-40^\circ,-30^\circ\right], d \in [70\text{m},80\text{m}]$ and D2D-RX$_2$ within $\theta \in \left[65^\circ,75^\circ \right], d \in \left[70\text{m},80\text{m}\right]$.
The distance between D2D-RX$_m$ and D2D-TX$_m$ is $20$m.
Fig. \ref{location} shows the visualization of the system geometry.

According to \cite{He2023} and \cite{Park2021}, the channel coefficients are given as follows:
$\mathbf{h}_{\text{BS}, k}=\sqrt{\varepsilon_0 d_{\text{BS},k}^{-\nu}}\sqrt{N_t}\mathbf{a}_t(\theta_k)$,
$\mathbf{h}_{\text{BS},m}=\sqrt{\varepsilon_0 d_{BS,m}^{-\nu}}\sqrt{N_t}\mathbf{a}_t(\theta_m)$,
$\mathbf{h}_{m,\text{BS}}^{\text{TX}}=\sqrt{\varepsilon_0 d_{m,\text{BS}}^{-\nu}}\sqrt{N_r}\mathbf{a}_r(\theta_m)$, $h_{m, m}^{\text{TX}}=\sqrt{\varepsilon_0 d_{m,m}^{-\nu}}$,
$h_{m',m}^{\text{TX}}=\sqrt{\varepsilon_0 d_{m', m}^{-\nu}}$,
$h_{m,k}^{\text{TX}}=\sqrt{\varepsilon_0 d_{m,k}^{-\nu}}$,
where $d_{\text{BS},k}$ denotes the distance between the BS and the $k$th CU,
$d_{m, \text{BS}}$ is the distance between D2D-TX$_{m}$ and the BS,
$d_{m', m}$ is the distance between $\{$D2D-TX$_{m'}\}_{m'\neq m}$ and D2D-RX$_{m}$,
$d_{m, k}$ is the distance between the D2D-TX$_{m}$ and the $k$th CU,
$\varepsilon_0$ denotes the path loss at a reference distance of one meter and is set to $\varepsilon_0=-30\mathrm{dB}$,
and $\nu$ denotes path loss exponent and set to $\nu=3$.

We compare the proposed ISAC scheme with four benchmark schemes: Communication-only scheme,  sensing-only scheme, the maximal ratio transmission (MRT) scheme\cite{Wu201911} and the zero forcing (ZF) scheme\cite{Lin2024}.

\textit{1) Communication-only scheme:} This scheme concentrates on designing a communication-oriented system without taking radar sensing into consideration. Accordingly, a communication-oriented optimization problem is constructed for comparing the ISAC scheme.
Specifically, the communication-only optimization problems is formulated by removing the radar SINR constraint while maintaining the same objective functions as those in (\ref{eq19}).
This benchmark scheme aids in evaluating the impact of integrating radar sensing functionality.

\textit{2) Sensing-only scheme:} In contrast to the communication-only scheme, the sensing-only scheme emphasizes a sensing-oriented system design and disregards the communication functionality.
 For formulating the sensing-only optimization problem compared with (\ref{eq19}), the radar SINR maximization problem is considered with constraints on the transmit power of the BS and D2D-TXs, which is given by
\begin{subequations}\label{eq36}
\begin{align}
\underset{\mathbf{w}_{k},\mathbf{u},p_{m}}
{\mathrm{max}}&\frac{\mathbf{u}^{H}\mathbf{A}(\theta_{0})\mathbf{F}\mathbf{A}^{H}(\theta_{0})\mathbf{u}}{\mathbf{u}^{H}\left(\mathbf{B}\mathbf{F}\mathbf{B}^{H}+\sigma_{r}^{2}\mathbf{I}_{N_{r}}\right)\mathbf{u}}\label{33a}\\
\text{s.t.}~~~~~&\sum_{k=1}^{K}\|\mathbf{w}_{k}\|^{2}\leq P_{\mathrm{BS}};\label{33c}\\
&0\leq p_{m}\leq P_{m}, \forall m;\label{33d}
\end{align}
\end{subequations}
The corresponding solution can be referred to \cite{Wu2018} or the fractional programming (FP) \cite{Su2022}.
This benchmark scheme can help us evaluate the impact of integrating communication functionality.

\textit{3) MRT scheme:} The MRT scheme aims to maximize the desired signal strength of the CUs, where the beamforming vector $\mathbf{w}_{k}$ is given by
\begin{equation}\label{eq37}
 \mathbf{w}_{k} = \sqrt{P_{k}}\dfrac{\mathbf{h}_{\text{BS},k}}{\|\mathbf{h}_{\text{BS},k}\|_2}.
\end{equation}

\textit{4) ZF scheme:} In ZF scheme, the beamforming vector $\mathbf{w}_k$ is forced to lie in the null space of channels  $\{\mathbf{h}_{\text{BS},m}\}_{m=1}^{M}$ and $\{\mathbf{h}_{\text{BS},k}\}_{k'\neq k}$. The beamforming vector $\mathbf{w}_{k}$ is given by
\begin{equation}\label{eq38}
 \mathbf{w}_{k} = \sqrt{P_{k}}\dfrac{\mathcal{P}_{\mathbf{H}_{k}}^{\perp}\mathbf{h}_{\text{BS},k}}{\|\mathcal{P}_{\mathbf{H}_{k}}^{\perp}\mathbf{h}_{\text{BS},k}\|_2}
\end{equation}
 where $\mathcal{P}_{\mathbf{H}_{k}}^{\perp}=\mathbf{I}_{N_t}-\mathbf{H}_{k}(\mathbf{H}_{k}^{H}\mathbf{H}_{k})^{-1}\mathbf{H}_{k}^{H}$ denotes the orthogonal projector of $\mathbf{H}_{k}$, with $\mathbf{H}_{k}\in\mathbb{C}^{N_t\times (M+K-1)}$ containing the channel vectors $\{\mathbf{h}_{\text{BS},m}\}_{m=1}^{M}$ and $\{\mathbf{h}_{\text{BS},k}\}_{k'\neq k}$ stacked in the columns.

 \begin{figure}[!t]
\centerline{\includegraphics[width=2.8in]{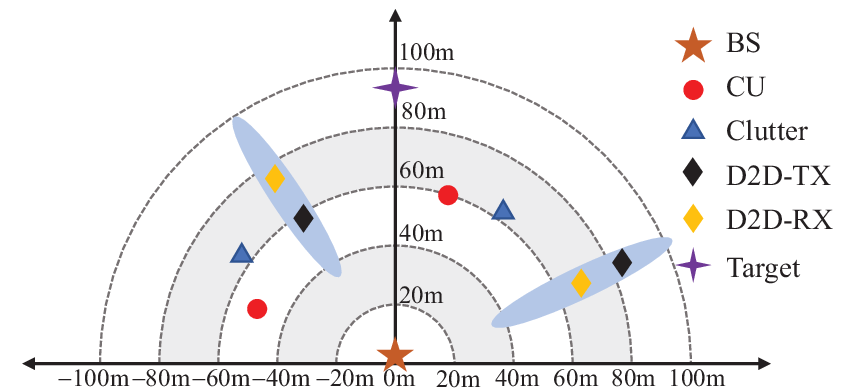}}
\caption{System geometry.}
\label{location}
\end{figure}


\begin{table}[!t]
	\centering
\caption{Ratio of Rank-one Solutions of $\mathbf{W}_k$, $k=1,2$}
	\begin{tabular}{lp{2cm}p{1cm}}

		\toprule  
		   &$\mathbf{W}_{1}$ & $\mathbf{W}_{2}$\\
		\midrule  
		\textbf{Proposed scheme}    & 99.2\% & 99.7\%\\
		\bottomrule 
	\end{tabular}
       \label{rankone}
\end{table}
The rank one constraints on $\mathbf{W}_k$ have been omitted in the SCA-based iterations,
and we evaluate the ratio of rank-one solutions of $\mathbf{W}_k$ in the iterations.
In the simulations, the radar SINR takes value randomly in a range, i.e., $\gamma_r\in [10, 20]\text{dB}$,
the maximum transmit power of D2D-TXs  $P_m\in [0.7, 20]\text{dBm}$,
and the maximum transmit power of the BS is also randomly selected, i.e., $P_\text{BS}\in [20, 120]\text{dBm}$.
We regard $\mathbf{W}_k$ as a rank one solution if its eigenvalues satisfy $\lambda_{\mathrm{lar}}/\lambda_{\mathrm{sec}}>10^5$, where $\lambda_{\mathrm{lar}}$ and $\lambda_{\mathrm{sec}}$ denote the largest eigenvalue and the second largest eigenvalue of $\mathbf{W}_k$, respectively.
We collect all $\mathbf{W}_k$s ($k = 1,2$) in the iterations with 2 CUs and 2 D2D pairs.
Table \ref{rankone} presents the ratios of rank-one solutions in the iterations in solving (\ref{eq34}).
It shows that the rank-one solution ratio of  $\mathbf{W}_k$  exceeds $99.2\%$.

\begin{figure}[!t]
\centerline{\includegraphics[width=2.8in]{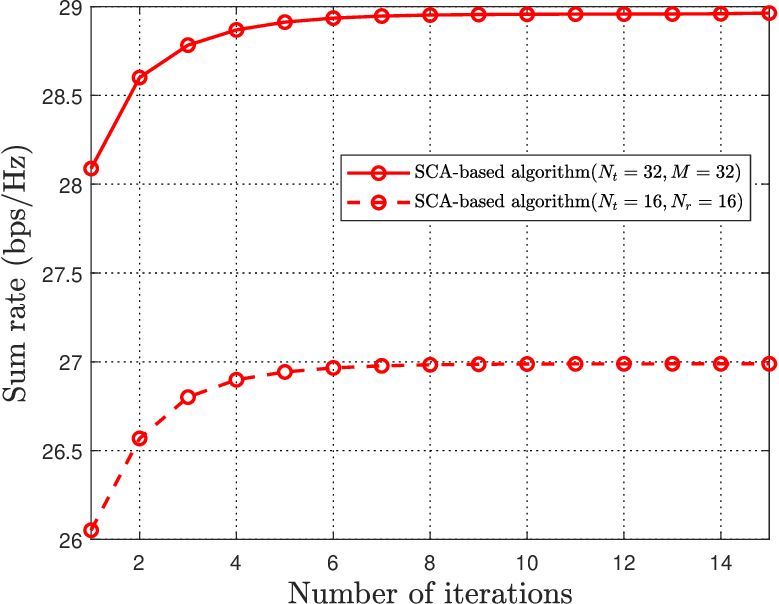}}
\caption{Convergence performance of the proposed scheme.}
\label{tc_diffiter}
\end{figure}

Fig. \ref{tc_diffiter} evaluates the convergence of the SCA-based algorithm for the proposed scheme with 2 CUs and 2 D2D pairs.
Both $16$ ($N_t =16$, $N_r = 16$) and $32$ ($N_t =32$, $N_r = 32$) antennas at the BS are considered.
It shows that the SCA-based algorithm converges after only $5$ iterations.

 \begin{figure}[!t]
\centerline{\includegraphics[width=2.8in]{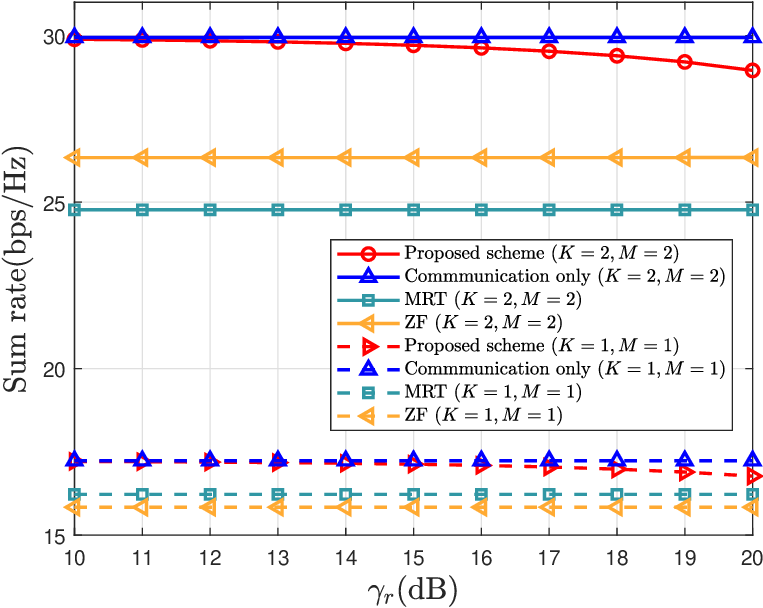}}
\caption{Sum rate at different SINR thresholds of radar sensing.}
\label{tc_diffrad}
\end{figure}




Next, we evaluate  sum rate and the SINR of radar sensing for various schemes.
Fig. \ref{tc_diffrad} shows the sum rate at different radar SINR thresholds $\gamma_r$ with $\gamma_m=\gamma_k=15\mathrm{dB}$.
It can be observed that integrating radar functionality into communication experiences slight performance loss compared with the communication-only scheme. Moreover, the proposed scheme significantly outperforms the communication oriented schemes MRT and ZF.
Specifically, when $\gamma_r=10$dB, The proposed scheme achieves almost the same sum rate as the communication only scheme, and the sum rate of the proposed scheme is $3.65$bps/Hz higher than ZF scheme when $K=2$ and $M=2$.
It is mainly because neither of the MRT and ZF schemes has the mechanism of suppressing interferences and strengthening the desired signal simultaneously.


 \begin{figure}[!t]
\centerline{\includegraphics[width=2.8in]{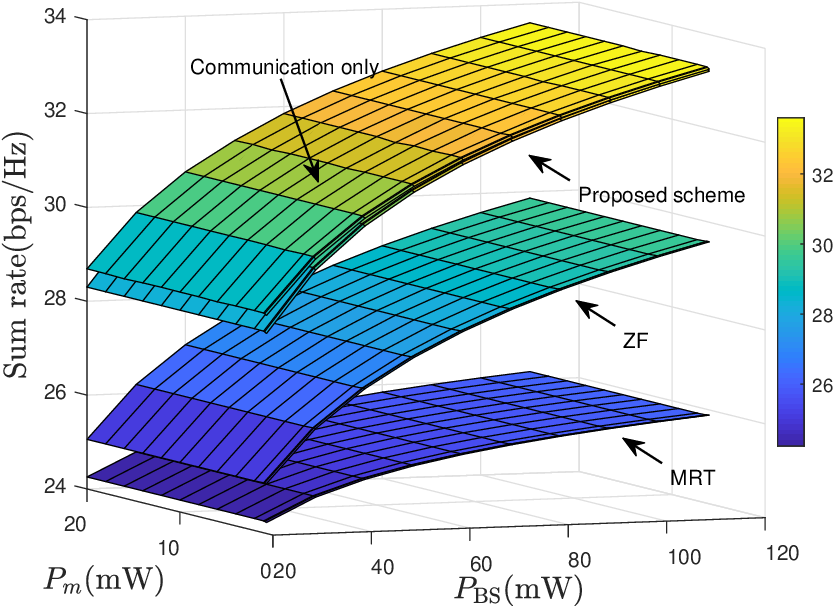}}
\caption{Sum rate versus maximum power constraints of BS  and D2D-TXs.}
\label{tctctc}
\end{figure}

 \begin{figure}[!t]
\centerline{\includegraphics[width=2.8in]{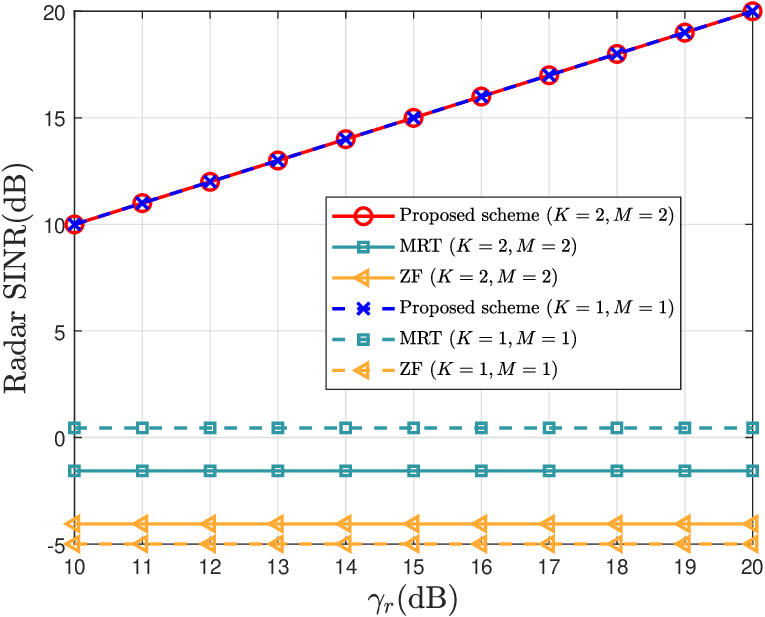}}
\caption{SINR for radar sensing of various schemes.}
\label{tc_radSINR}
\end{figure}

 \begin{figure}[!t]
\centerline{\includegraphics[width=2.8in]{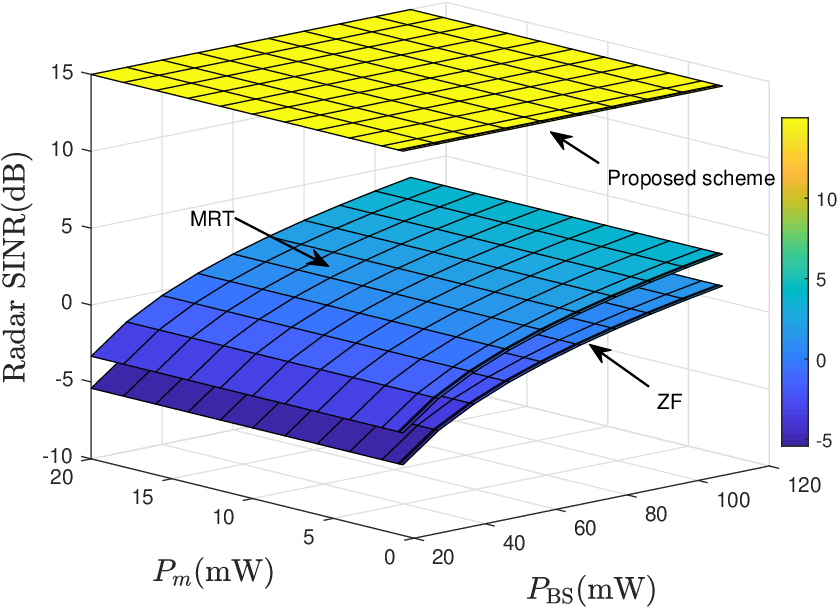}}
\caption{Radar SINR of various schemes at different SINR thresholds of D2D communication $\gamma_m$ and CU communication $\gamma_k$.}
\label{tc_diffpradSINR}
\end{figure}

Fig. \ref{tctctc} shows the sum rate at  different maximum power constraints of the BS $P_\mathrm{BS}$  and D2D transmitters  $P_m$.
It can be observed that the sum rate increases with $P_\mathrm{BS}$ and $P_m$,
and allocating more power to the BS leads to quicker increase of the sum rate than allocating power to D2D transmitters.
It can also be seen that the proposed scheme outperforms  the MRT and ZF schemes.
The minimal gap between the proposed scheme and ZF is above $3.2$bps/Hz, and the gaps between the proposed scheme and MRT are more pronounced.
Moreover, the gap between the proposed scheme and the communication-only scheme is small, implying that radar sensing can be effectively integrated into the communication system with slight performance loss in sum rate.

Fig. \ref{tc_radSINR} shows the  SINR of radar sensing for different schemes.
The radar SINR of the proposed scheme equals the required $\gamma_r$, while MRT and ZF have constant radar SINR at different $\gamma_r$.
This is because MRT and ZF employ fixed power allocation strategy.
It can be seen from Figs. \ref{tc_diffrad} and \ref{tc_radSINR} that our proposed scheme
outperforms MRT and ZF in both sum rate (communication) and radar SINR (sensing).
Fig. \ref{tc_diffpradSINR} shows the radar SINR at different maximum power constraints of the BS $P_\mathrm{BS}$  and D2D transmitters $P_m$.
Again, we can see that our proposed scheme achieves higher radar SINR than MRT and ZF.

\begin{figure*}[!t]
  \centering
  \includegraphics[width=7.0in]{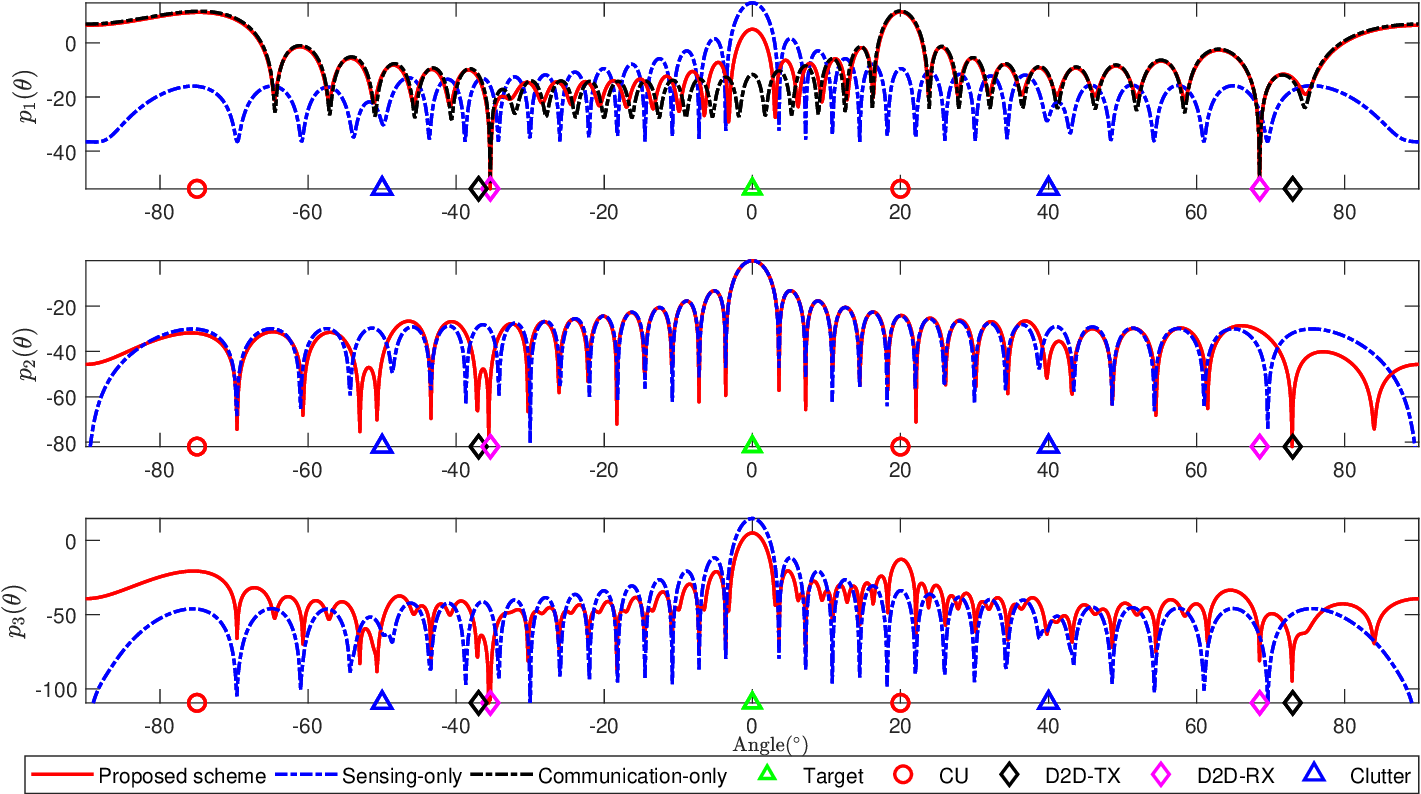}
  \caption{Transmit beampattern $p_1(\theta)$, receive beampattern $p_2(\theta)$ and combined beampattern $p_3(\theta)$ at the BS.}
  \label{beamforming2new}
\end{figure*}

Finally, we evaluate the beampatterns at the BS when $\gamma_r=15\mathrm{dB}$, $P_{\text{BS}}=30$mW, and $P_{m}=10$mW.
Fig. \ref{beamforming2new} shows the transmit beampattern $p_1(\theta)$, the receive beampattern $p_2(\theta)$, and the combined beampattern $p_3(\theta)$
 where
\begin{equation}\label{eq39}
p_1(\theta) = \mathbb{E}\left(\left|\mathbf{a}_t^H(\theta)\mathbf{x}\right|^2\right),
\end{equation}
\begin{equation}\label{eq40}
p_2(\theta) = \dfrac{\left|\mathbf{u}^{H}\mathbf{a}_r(\theta)\right|^2}{\mathbf{u}^{H}\mathbf{u}},
\end{equation}
\begin{equation}\label{eq41}
p_3(\theta) = \dfrac{\mathbb{E}\left(\left|\mathbf{u}^{H}\mathbf{a}_r(\theta)\mathbf{a}_t^H(\theta)\mathbf{x}\right|^2\right)}{\mathbf{u}^{H}\mathbf{u}}.
\end{equation}
It is observed from the transmit beampattern $p_1(\theta)$ in Fig. \ref{beamforming2new} that  the sensing-only scheme focuses the power on the target rather than the CUs, and  the communication-only scheme focuses on only the CUs, while the proposed scheme focuses on both the target and CUs.
Moreover, both the proposed scheme and the communication-only scheme have null steering beam at D2D-Rxs, reducing the interferences from the BS to the D2D communications.
However, the sensing-only scheme may interfere D2D communications significantly.
Next we examine the receive beampattern $p_2(\theta)$ in Fig. \ref{beamforming2new} of the proposed and sensing-only schemes.
Note that the communication-only scheme has no  receive beampattern.
It can be seen from the receive beampattern $p_2(\theta)$  that the proposed scheme has better suppressing null beam from the angles of clusters and the D2D-Txs than the sensing-only scheme.
The combined beampattern $p_3(\theta)$ shows that the proposed scheme can effectively suppress the interferences between the BS and D2D users, clusters and focus the power of target and CUs.


\section{Conclusion}
This work has investigated joint transceiver beamformer design and power allocation in an ISAC-enabled D2D underlaid cellular network, where the full-duplex BS performs the radar sensing and downlink CU communication, and D2D pairs engage in direct communication in the same time-frequency resource.
A sum rate maximization problem is formulated to improve the spectral efficiency. Accordingly,
an SCA-based algorithm is proposed to solve the optimization problem.
Extensive numerical results demonstrate that the proposed scheme can achieve satisfactory performance for radar sensing, CU and D2D communications, compared with the state-of-art schemes.

\appendix
\subsection{Proof of Lemma 1}
The numerator of (\ref{eq14}) can be equivalently transformed into $\mathbf{a}_{t}^{H}\mathbf{F}\mathbf{a}_{t}|\mathbf{u}^{H}\mathbf{a}_{r}|^{2}$,
and (\ref{eq14}) can be rewritten as
\begin{equation}\label{37}
\begin{aligned}
\underset{\mathbf{u}}
{\mathrm{max}}~~\frac{|\mathbf{u}^{H}\mathbf{a}_{r}(\theta_{0})|^{2}}{\mathbf{u}^{H}\mathbf{G}\mathbf{u}}\\
\end{aligned}
\end{equation}
As scaling $\mathbf{u}$ with an arbitrary constant will not change the objective function value, we place a constraint $\mathbf{u}^{H}\mathbf{a}_{r}=1$.
Therefore, problem (\ref{37}) can be recast as a minimum variance distortionless response (MVDR) problem as
\begin{equation}\label{38}
\begin{aligned}
\underset{\mathbf{u}}
{\mathrm{min}}~~&\mathbf{u}^{H}\mathbf{G}\mathbf{u},\\
\text{s.t.}~~~&\mathbf{u}^{H}\mathbf{a}_{r}(\theta_{0})=1.
\end{aligned}
\end{equation}
The Lagrangian function of problem (\ref{38}) is given by
 \begin{equation}\label{39}
\begin{aligned}
\mathcal{L}\left(\mathbf{u};\lambda\right)=\mathbf{u}^{H}\mathbf{G}\mathbf{u}+\lambda\left(\mathbf{u}^{H}\mathbf{a}_{r}(\theta_{0})-1\right)
\end{aligned}
\end{equation}
where $\lambda$ denotes the lagrange multiplier.
According to the Karush-Kuhn-Tucker (KKT) condition, the optimal solution should satisfy
$\frac{\partial \mathcal{L}}{\partial \mathbf{u}} = \mathbf{0}$ and $\mathbf{u}^{H}\mathbf{a}_{r}(\theta_{0})=1$.
Thus, $\mathbf{u}$ is given by
 \begin{equation}\label{40}
\begin{aligned}
\mathbf{u} = \frac{\mathbf{G}^{-1}\mathbf{a}_{r}(\theta_{0})}{\mathbf{a}_{r}^{H}(\theta_{0})\mathbf{G}^{-1}\mathbf{a}_{r}(\theta_{0})}.
\end{aligned}
\end{equation}

\normalem

\end{document}